\documentclass[aps,pra,twocolumn,superscriptaddress,amsmath,10pt]{revtex4}
\usepackage{amsthm,amssymb,graphicx}
\usepackage{bm,mathrsfs,dsfont,nicefrac}
\usepackage{braket,comment}
\usepackage{appendix}
\usepackage{xcolor}
\usepackage{lscape}
\usepackage[normalem]{ulem}
\usepackage{stmaryrd}
\theoremstyle{remark}
\usepackage{algorithm, algpseudocode}
\usepackage{algorithmicx}

\newtheorem{proposition}{Proposition}

\begin{document}
\title{Quantum capacity of bosonic dephasing channel}
\author{Amir Arqand}
\affiliation{%
	Department of Physics, Sharif University of Technology, Tehran, Iran}
\author{Laleh Memarzadeh}
\affiliation{%
	Department of Physics, Sharif University of Technology, Tehran, Iran}
\author{Stefano Mancini}
\affiliation{%
School of Science \& Technology, University of Camerino, I-62032 Camerino, Italy}
\affiliation{
INFN Sezione di Perugia, I-06123 Perugia, Italy}
\date{\today}
\begin{abstract}
We study the quantum capacity of continuous-variable dephasing channel, which is a notable
example of a non-Gaussian quantum channel. We prove that a single letter formula applies. 
 The optimal input state is found to be diagonal in the Fock basis and with a distribution that
is a discrete version of a Gaussian. Relations between its mean/variance and dephasing rate/input
energy are put forward. We then show that by increasing the input energy, the capacity saturates to a finite value. We also show that it decays exponentially for large values of dephasing rates.
\end{abstract}

\maketitle


\section{Introduction}

Any physical process can be regarded as a quantum channel, i.e. a stochastic map 
on the space of states that causes a state change.
As such it can be characterized by its ability in conveying information.
A notable example is provided by the quantum capacity of a channel, that shows its ability to transfer unaltered the entanglement of the input system with a reference system \cite{MW20}.

Finding the quantum capacity of a channel is challenging because not only the optimization of 
an entropic functional (coherent information) over input states is required, but also its regularization \cite{Qcomplexity,Qcomplexity2,Qcomplexity3}. 
This task becomes even harder when dealing with infinite-dimensional (so-called continuous-variable) systems. That is why till now, in this framework, the attention has been confined to Gaussian channels, i.e. maps that transform Gaussian states into Gaussian states \cite{Gchannel,Gchannel2}. 
For instance, the coherent information of the lossy channel (a special case of Gaussian channels) is known to be additive and hence its quantum capacity is computed \cite{WPG07} (see \cite{WQ16} for the general formalism of energy-constrained quantum capacity).
For more general Gaussian channels, an upper bound of quantum capacity can be obtained by evaluating one-shot coherent information of the channel \cite{HW01}.
 
Nevertheless, there is an increasing pressure to go beyond the Gaussian channels paradigm \cite{MM,GHLP20}.  
Heading in this direction, we investigate here the quantum capacity of one of the most physically relevant 
non-Gaussian channels, namely the dephasing channel (see e.g. \cite{WM}). 
It causes the reduction of the off-diagonal terms in the Fock basis, 
thus washing out coherence properties of the state.
This happens for instance with uncertainty path length in optical fibers \cite{Derickson}.

Here we prove that for dephasing channel the single letter formula applies for the quantum capacity. 
The optimal input state is found to be a non-Gaussian state, which is diagonal in the Fock basis and with a distribution
that is a discrete version of a continuous Gaussian distribution. In fact we show that the optimal probability distribution is a symmetric unimodal probability distribution. Then we take discrete Gaussian probability distribution as an ansatz and show that by proper selection of mean value and variance, it becomes a perfect fit for optimal probability distribution.
 The relation between optimal mean/variance and dephasing rate/input energy are put forward.
Finally, we show that by increasing input energy, the quantum capacity saturates to a finite value which depends on the noise parameter of the channel. We also
 show that, for a large value of dephasing rate, the quantum capacity decays exponentially with dephasing rate.

The structure of the paper is as follows: In section \ref{sec:backgroun} we set our notation and explain the terms we need for our next purposes. Section \ref{sec:QuantumDephasing} is for short review on the quantum dephasing channel and its different representations that we are going to use in proceeding sections. Section \ref{sec:QuantumCapacity} is devoted to the quantum capacity
of the dephasing channel, containing analytical results for proving that single letter formula applies and showing the structure of the optimal input state. In section \ref{sec:numerical} we introduce our approach for using the replica method 
to numerically evaluate quantum capacity. Its asymptotic behavior is then discussed in section \ref{sec:asymptotic}. Finally, section \ref{sec:conclusion} concludes with a summary and discussion of the results.


\section{notation and preliminaries}
\label{sec:backgroun}

In this section we set our notation and review relevant concepts and terms used in the proceeding sections for deriving the quantum capacity of bosonic dephasing channel.  
Here, states of the initial system and environment respectively belong to Hilbert-spaces denoted by $\mathcal{H}_S$ and $\mathcal{H}_E$. Similarly, the Hilbert-space of the final system and environment are respectively denoted by $\mathcal{H}_{S'}$ and $\mathcal{H}_{E'}$.
Density operators on Hilbert-space $\mathcal{H}$, belong to $\mathcal{T}(\mathcal{H})$, the set of linear positive operators on $\mathcal{H}$ with trace one.

Due to the unavoidable interaction between system and environment, described by an isometry $U:\mathcal{H}_{S}\otimes \mathcal{H}_{E}\to\mathcal{H}_{S'}\otimes \mathcal{H}_{E'}$, there are noise effects on the system. The most general form of a system evolution is then given by a completely positive trace preserving (CPTP) map or a quantum channel
$\mathcal{N}:\mathcal{T}(\mathcal{H}_S)\to\mathcal{T}(\mathcal{H}_{S'})$ described by tracing over environment degrees of freedom of the final state, that is 
\begin{equation}
    \mathcal{N}(\rho)=\operatorname{Tr}_{E'}\left(U(\rho\otimes\ket{0}\bra{0})U^\dagger\right).
\end{equation}
 Here $\rho\in\mathcal{T}(\mathcal{H}_S)$ is any initial state of the system and $\ket{0}\in\mathcal{H}_E$ a fixed initial state of the environment. 
For each channel $\mathcal{N}$, its complementary channel $\mathcal{N}^c$ is a CPTP map $\mathcal{N}^c:\mathcal{T}(\mathcal{H}_S)\to\mathcal{T}(\mathcal{H}_{E'})$ given by
\begin{equation}
\label{eq:complementary}
    \mathcal{N}^c(\rho)=\operatorname{Tr}_{S'}\left(U(\rho\otimes\ket{0}\bra{0})U^\dagger\right).
\end{equation}
Two important properties of quantum channels which we use here are degradablity and entanglement breaking. A channel $\mathcal{N}:\mathcal{T}(\mathcal{H}_S)\to\mathcal{T}(\mathcal{H}_{S'})$, is degradable if there exists a channel 
$\mathcal{M}:\mathcal{T}(\mathcal{H}_{S'})\to\mathcal{T}(\mathcal{H}_{E'})$ such that
\begin{equation}
    \mathcal{M}\circ\mathcal{N}=\mathcal{N}^c,
\end{equation}
where $\circ$ denotes composition of maps \cite{DevetakShor2005}. To recall the definition of entanglement breaking channel, first we denote the Hilbert space of a reference state by 
$\mathcal{H}_R$ and the identity operator on it by $\mathds1_R$. Then,
a quantum channel $\mathcal{N}:\mathcal{T}(\mathcal{H}_S)\to\mathcal{T}(\mathcal{H}_{S'})$ is entanglement breaking 
if the map $\mathds1_R\otimes\mathcal{N}: \mathcal{T}(\mathcal{H}_R\otimes\mathcal{H}_S)\to\mathcal{T}(\mathcal{H}_R\otimes\mathcal{H}_{S'})$ maps every density operator to a separable state \cite{Holevo_eb,Horodecki_eb}.
 Actually it is known that if $\mathds1_R\otimes\mathcal{N}$ maps  a maximally entangled state to a separable state, $\mathcal{N}$ is entanglement breaking \cite{Holevo_eb,Horodecki_eb}.

 The quantum capacity of a channel $\mathcal N$,  is the highest rate of reliable quantum information transmission through the channel. It can be expressed in terms of the coherent information of the channel's output state. The latter quantity is defined as
\begin{equation}
J(\rho,{\cal N}) \equiv S\left({\cal N}(\rho)\right)-S\left(\mathcal{N}^c(\rho)\right),
\end{equation}
with $S(\rho)=-\text{Tr}\left(\rho\log\rho\right)$ the von Neumann entropy of $\rho$ (throughout the paper we use logarithm to base 2). 
 Then, the quantum capacity of the $\mathcal{N}$ results as the regularized maximum coherent information of the output of infinitely many channel's uses \cite{MW20}, that is
\begin{equation}
\label{Q}
Q(\mathcal{N})=\lim_{n\to\infty}\frac{1}{n} \left[
\max_{\rho^{(n)}}
J\left(\rho^{(n)},{\cal N}^{\otimes n}\right)
\right],
\end{equation}
with maximization over all density operators 
$\rho^{(n)}\in\mathcal{T}(\mathcal{H}_S^{\otimes n})$. 
For channels with additive coherent information, maximizing the coherent information of single channel use over density operators on $\mathcal{H}_S$ is sufficient for computing the quantum capacity. Hence the formula \eqref{Q} can be simplified to single-letter expression \cite{DevetakShor2005}: 
\begin{equation}
    \label{eq:Q1}
    Q(\mathcal{N})=\max_{\rho}
J\left(\rho,{\cal N}\right).
\end{equation}
For degradable channels, the coherent information is additive \cite{DevetakShor2005}. Furthermore, a channel $\mathcal{N}$ is degradable if an only if its complementary channel 
$\mathcal{N}^c$ is entanglement breaking \cite{Cubitt}. 
Therefore, the quantum capacity of degradable channels, or channels with entanglement breaking complementary, is given by the single-letter formula in Eq.~(\ref{eq:Q1}).

\section{Quantum dephasing channel}
\label{sec:QuantumDephasing}
In this section we describe the quantum dephasing channel by giving various representations for it. Furthermore, we explain how such a channel is related to a Markovian process and forms a semi-group.

The continuous-variable quantum dephasing effect (see e.g. \cite{WM}) provides a notable example of a non-Gaussian channel. 
The channel $\mathcal {N}_\gamma:{\cal T}({\cal H}_S)\to{\cal T}({\cal H}_S)$ (note that here  ${\cal H}_S$ and ${\cal H}_{S'}$ are isomorphic) can be dilated into a single mode environment with the following unitary 
\begin{eqnarray}
\label{eq:U}
U&=&e^{-i\sqrt{\gamma} (a^\dag a)(b+b^\dag)}\notag\\
&=&e^{-i\sqrt{\gamma} (a^\dag a) b^\dag} 
e^{-i\sqrt{\gamma} (a^\dag a) b} e^{-\frac{1}{2}\gamma (a^\dag a)^2}.  
\end{eqnarray}
Here $a$ and $a^\dagger$ are bosonic ladder operators acting on system Hilbert space 
${\cal H}_S$, $b,b^\dag$ are bosonic ladder operators on the environment 
Hilbert space ${\cal H}_{E'}$
(isomorphic to ${\cal H}_E$ and $\mathcal{H}_S$), and $\gamma\in[0,+\infty)$ is a parameter that determines the dephasing rate.
The unitary evolution by Eq.~(\ref{eq:U}), can be represented as a controlled displacement gate in a quantum circuit. The system acts as a controlled mode prepared in the Fock basis and target -- environment -- mode experiences a displacement proportional to $\sqrt{\gamma}$.

For the system evolution, by tracing over environment degrees of freedom, 
we get
\begin{equation}
\label{eq:NU}
\rho\mapsto\mathcal{N}_{\gamma}(\rho)={\rm Tr}_E\left[U\left(\rho\otimes |0\rangle\langle 0|\right) U^\dag\right].
\end{equation}
with $\ket{0}$, the vacuum of the environment. If we expand the input state in the Fock basis $\rho=\sum_{m,n=0}^\infty \rho_{m,n}|m\rangle\langle n|$ the effect of ${\cal N}_\gamma$ reads
\begin{equation}
\label{eq:rhoout}
\rho\mapsto\mathcal{N}_{\gamma}(\rho)=\sum_{m,n}^\infty e^{-\frac{1}{2}\gamma(m-n)^2}\rho_{m,n} |m\rangle\langle n|,
\end{equation}
which clearly shows that the diagonal elements of the input are preserved, while the off diagonal ones tend to be washed out. The channel's output as given in Eq.~(\ref{eq:rhoout}), is also the solution of the following Markov master equation: 

\begin{equation}
\label{eq:mastereq}
    \dot{\rho}(t)=\mathcal{L}[\rho(t)],
\end{equation}
with 
\begin{equation}
\label{eq:Gen}
    \mathcal{L}[\bullet]\equiv 2(a^\dagger a)\bullet(a^\dagger a)
    -\left((a^\dagger a)^2\bullet\right)
    -\left(\bullet (a^\dagger a)^2\right),
\end{equation}
where $\bullet$ is any operator in $\mathcal{T}(\mathcal{H}_S)$ and 
$t\equiv\gamma$. Hence, the set of dephasing channels 
$\{{\cal N}_\gamma\}$ forms a semi-group under composition:  ${\cal N}_\gamma\circ{\cal N}_{\gamma'}
={\cal N}_{\gamma+\gamma'}$.

Kraus representation of the channel is given by
\begin{equation}\label{deph}
\rho \mapsto {\cal N}_\gamma(\rho)=\sum_{j=0}^\infty K_j \rho K_j^\dag,
\end{equation}
where Kraus operators $K_j=\bra{j}U\ket{0}$, with $\ket{j}\in\mathcal{H}_E$ being the number state in environment and unitary evolution $U$ as given in Eq.~(\ref{eq:U}), have the following explicit form:  
\begin{equation}
\label{Kdeph}
K_j= e^{-\frac{1}{2}\gamma (a^\dag a)^2} \frac{\left(-i\sqrt{\gamma} a^\dag a\right)^j}{\sqrt{j!}}.
\end{equation}
The channel's action can also be written as
\begin{equation}
\label{eq:intdep}
\rho\mapsto \mathcal{N}_{\gamma}(\rho)=\int_{-\infty}^{+\infty} e^{-ia^\dag a\phi} \rho e^{ia^\dag a\phi }\, p(\phi) d\phi,
\end{equation}
with
\begin{equation}
p(\phi)=\sqrt{\frac{\gamma}{2\pi}}e^{-\frac{1}{2} \gamma \phi^2}.
\label{pphi}
\end{equation}
This means a randomization of the phase $\phi$ according to the probability distribution \eqref{pphi}.
Note that  $\phi$ as a random variable must be defined on the sample space 
$\mathbb{R}$, not $[0,2\pi]$.\newline

\section{Quantum capacity}
\label{sec:QuantumCapacity}
In this section we derive the explicit form of the complementary channel of bosonic dephasing channel \eqref{eq:NU},\eqref{eq:rhoout}. We show that the quantum capacity of the latter is given by the single-letter formula (\ref{eq:Q1}). Based on this we derive the structure of the optimal input state.

For the channel $\mathcal{N}_\gamma$ in Eq.~(\ref{eq:NU}), the complementary channel (\ref{eq:complementary})
$\mathcal{N}^c_\gamma:{\cal T}({\cal H}_S)\to{\cal T}({\cal H}_E)$ is given by

\begin{align}
\label{eq:compl}
\rho\mapsto \mathcal{N}^c_\gamma(\rho)&= {\rm Tr}_S\left[U\left(\rho\otimes \ket{0}\bra{0}\right) U^\dag\right]\notag\\
&= {\rm Tr}_S\left[\sum_{m,n}\rho_{m,n} \ket{m}\bra{n}\otimes \ket{-i\sqrt{\gamma} m}\bra{-i\sqrt{\gamma} n}\right] \notag\\
&= \sum_{m}\rho_{m,m} \ket{-i\sqrt{\gamma} m}\bra{-i\sqrt{\gamma} m}\cr
&=e^{-\frac{i\pi}{2}a^\dagger a }\left(\sum_{m}\rho_{m,m} |\sqrt{\gamma} m\rangle\langle \sqrt{\gamma} m|\right)e^{i\frac{\pi}{2}a^\dagger a},\cr
\end{align}

where $U$ is defined in Eq.~(\ref{eq:U}) and $|\sqrt{\gamma} m\rangle$ is a coherent state of real amplitude $\sqrt{\gamma} m$, 
i.e.
\begin{equation}
\label{eq:coherent}
|\sqrt{\gamma} m\rangle=e^{-\gamma m^2/2}\sum_{k=0}^{\infty}\frac{(\sqrt{\gamma} \, m)^k}{\sqrt{k!}} |k\rangle.
\end{equation}
The complementary channel is a mixture of coherent state. In fact the input state with $m$ photon number is projected into a coherent state with an amplitude proportional to $m$. The ultimate output state of the complementary channel is a mixture of these coherent states, with the weight given by the probability of having $m$ photons in the input\cite{Horodecki_eb}. 
The complementary channel $\mathcal{N}_\gamma^c$ (\ref{eq:compl}) is entanglement breaking. To show this 
we consider a two-mode squeezed vacuum state
\begin{equation}
|\Psi\rangle_{RS}=\sum_{n=0}^\infty \lambda^n |n\rangle_R |n\rangle_S, \qquad
0\leq \lambda \leq 1,
\end{equation}
being $R$ a reference system isomorphic to $S$ and $\lambda$ the squeezing parameter.
Using Eq.\eqref{eq:compl}, we immediately arrive to 

\begin{align}
&\left(\mathds1_R\otimes\mathcal{N}^c_\gamma\right) |\Psi\rangle_{RS}\langle\Psi|\cr
&=\sum_{m}\lambda^{2m} \ket{m}_R\bra{m}\otimes\ket{-i\sqrt{\gamma} m}_E\bra{-i\sqrt{\gamma} m},
\end{align}

which is a mixture of product states and hence is a separable state for any value of $\lambda$ . Being $\mathcal{N}^c_{\gamma}$ entanglement breaking, $\mathcal{N}_\gamma$ is degradable. Therefore, according to \S\ref{sec:backgroun}, the quantum capacity of bosonic dephasing channel is given by single-letter formula \ref{eq:Q1}: 
\begin{equation}\label{eq:Qsingle}
Q(\mathcal{N}_\gamma)=\max_{\rho} J\left(\rho,{\cal N}_\gamma\right).
\end{equation}

Next we use the phase-covariance property of $\mathcal{N}_{\gamma}$ and the concavity of the coherent information, to restrict the set of density operators over which the maximization in Eq.~(\ref{eq:Qsingle}) should be performed. Similar argument is used in the context of bosonic pure-loss channels \cite{Kyungjoo}. 
\begin{proposition}
\label{prop:OptimalInput}
The optimal input state to ${\cal N}_\gamma$ for the quantum capacity \eqref{eq:Qsingle} is diagonal in the Fock basis.
\end{proposition}
\begin{proof}
From Eq.~(\ref{eq:intdep}) it follows that the quantum dephasing channel is phase-covariant, that is for $U_\theta=e^{-i a^\dag a \theta}$ with $\theta\in[0,2\pi)$ we have 
\begin{equation}
\label{eq:Ncovariant}
{\cal N}_\gamma(U_\theta\rho U^\dag_\theta)
=U_\theta{\cal N}_\gamma(\rho) U_\theta^\dag.
\end{equation}
Similarly, from Eq.~(\ref{eq:compl}), we conclude that also the complementary channel 
$\mathcal{N}^{c}_{\gamma}$ is phase-covariant:
\begin{equation}
\label{eq:NCcovariant}
\mathcal{N}^{c}_{\gamma}(U_\theta\rho U^\dag_\theta)
=U_\theta\mathcal{N}^{c}_{\gamma}(\rho) U_\theta^\dag.
\end{equation}
As the von-Neumann entropy is invariant under unitary conjugate, from Eqs.~(\ref{eq:Ncovariant}) and (\ref{eq:NCcovariant}) we conclude that
\begin{equation}
\label{eq:vonNumInv}
J(\rho_\theta,{\cal N}_\gamma)=J(\rho,{\cal N}_\gamma),
\end{equation}
with 
$\rho_\theta\equiv U_\theta\rho {U_\theta}^{\dag}$. 
On the other hand, for degradable channels, the coherent information is a concave function of its input state, that is
\begin{equation}
\label{Jcon}
\int_0^{2\pi} J\left(\rho_\theta,{\cal N}_\gamma\right)p(\theta)d\theta
\le J\left(\int_0^{2\pi} \rho_\theta p(\theta) d\theta ,{\cal N}_\gamma\right),
\end{equation}
for any probability distribution $p(\theta)$. 
Thus from Eq.~(\ref{eq:vonNumInv}) and \eqref{Jcon} it is straightforward to see that 
\begin{equation}
J\left(\rho,{\cal N}_\gamma\right)
\le J\left(\int_0^{2\pi} \rho_\theta p(\theta) d\theta ,{\cal N}_\gamma\right).
\label{Jcon2}
\end{equation}
Then, choosing $p(\theta)$ as flat distribution, we find
\begin{align}
\int_0^{2\pi} \rho_\theta p(\theta) d\theta&=\frac{1}{2\pi}\sum_{m,n}^\infty \int_0^{2\pi}
\rho_{m,n} |m\rangle\langle n| e^{i\theta(m-n)}d\theta \notag\\
&=\sum_{n}^\infty \rho_{n,n} |n\rangle\langle n|.
\end{align}
Finally, inserting this into the r.h.s. of \eqref{Jcon2} gives 
\begin{equation}
\label{eq:JUpperUniform}
J\left(\rho,{\cal N}_\gamma\right)
\le J\left(\sum_{n=0}^\infty \rho_{n,n} |n\rangle\langle n|, {\cal N}_\gamma\right),
\end{equation}
i.e. the desired result.
\end{proof}
The fact that the optimal input state is diagonal in the Fock basis, can be interpreted more intuitively by noting that the steady-states, or the state that remains invariant under the Markov process generated by $\mathcal{L}$ in Eq.~(\ref{eq:Gen}), 
is not unique. In fact all Fock states are invariant under the dynamics generated by $\mathcal{L}$ in Eq.~(\ref{eq:Gen}). Hence every mixture of invariant states, which is a state diagonal in the Fock basis, is invariant under this evolution. Of course in the subset of steady-states of Markovian dynamics generated by $\mathcal{L}$ in Eq.~(\ref{eq:Gen}) we should find the one that can carry the 
 largest amount of quantum information through the channel.

As a consequence of Proposition \ref{prop:OptimalInput}, the maximization in Eq.~(\ref{eq:Qsingle}) reduces to the maximization over classical probability distribution: 
\begin{align}
\label{eq:Qfinal}
Q({\cal N}_\gamma)=\max_{p_m} 
\Bigg[ &S\left(\sum_{m=0}^{\infty} p_m |m\rangle\langle m| \right)\notag\\
&-  S\left(\sum_{m=0}^{\infty} p_m |\sqrt{\gamma} m\rangle\langle \sqrt{\gamma}m| \right) \Bigg],
\end{align}
where we have used the last equality in Eq.~(\ref{eq:compl})
and the invariance of entropy under unitary conjugation. A lower bound to \eqref{eq:Qfinal} can be found by considering an input state to be 
diagonal in the Fock basis and containing only two elements with equal weight, i.e. 
\begin{equation}\label{eq:Omega2}
\Omega_j=\frac{1}{2}(\ket{n}\bra{n}+\ket{n+j}\bra{n+j}),
\end{equation}
where $n,j$ are arbitrary non-negative integers. 
In such a case it is easy to see that 
$\sum_m p_m |\sqrt{\gamma} m\rangle\langle \sqrt{\gamma}m|$ is diagonalized 
in the following basis 
\begin{align}
&\frac{1}{\sqrt{2+2e^{-\gamma j^2/2}}} \left(|\sqrt{\gamma} n\rangle+|\sqrt{\gamma} (n+j)\rangle\right), \\
&\frac{1}{\sqrt{2-2e^{-\gamma j^2/2}}} \left(|\sqrt{\gamma} n\rangle-|\sqrt{\gamma} (n+j)\rangle\right),
\end{align}
with eigenvalues
\begin{equation}
\label{eq:q}
q_\pm(j)\equiv\frac{1}{2}\left(1\pm e^{-\gamma j^2/2}\right).
\end{equation}
Thus we have 
\begin{equation}
\label{eq:JUni2}
J(\Omega_j,\mathcal{N})=1-H_2(q_+(j),q_-(j)),
\end{equation}
with $H_2$ the binary entropy.

We note that the eigenvalues of $\mathcal{N}^{c}_{\gamma}(\Omega_j)$ in Eq. (\ref{eq:q}) do not depend on $n$. Furthermore, by increasing $j$, the distance between $q_+(j)$ and $q_-(j)$ decreases and as a consequence $H_2(q_+(j),q_-(j))$ decreases too. Therefore,  $J(\Omega_j,\mathcal{N})$ in Eq.~(\ref{eq:JUni2}) is maximized for $j=1$ and a a lower bound for quantum capacity is given by
$J(\Omega_{1},\mathcal{N}_{\gamma})$ which is obtained for input state $\Omega_1$ with arbitrary $n$. 

In order to obtain the quantum capacity, it is necessary to go beyond 
the input state (\ref{eq:Omega2}) by considering more terms in the sum and  
non trivial probability distributions.
The task is complicated because computing the second term of Eq.~(\ref{eq:Qfinal}) requires the diagonalization of a mixture of infinite number of coherent states. Hence, in the next section we will use numerical tools. 


\section{Numerical analysis}
\label{sec:numerical}

In this section we resort to numerical techniques to evaluate the quantum capacity. 
First we truncate the space ${\cal H}_S$ to dimension $N+1$. According to Proposition \ref{prop:OptimalInput}, the optimal input state is diagonal in the Fock basis and in
a truncated Hilbert space it takes the form
\begin{equation}
    \rho=\sum_{m=0}^N p_m\ket{m}\bra{m}.
\end{equation}
The mean energy of this state is given by $\sum_{m=0}^Nmp_m$, with maximum value $N$. Therefore the truncation of Hilbert space
can be regarded as constraining the input average energy \footnote{On the other way around, if we constraint the input average energy as $\sum_{m=0}^{\infty}mp_m=E<+\infty$,
we know that $\forall\epsilon>0, \; \exists N_\epsilon : |\sum_{m=0}^{N}mp_m-E|<\epsilon$ for $N>N_\epsilon$. And this can be regarded as truncating the Hilbert space to the dimension $N_\epsilon$ (within an accuracy $\epsilon$).}.
Then, we find the following maximum numerically
\begin{align}
\label{eq:QTrunckated}
Q_{N+1}({\cal N}_\gamma)=\max_{p_m} 
\Bigg[ &S\left(\sum_{m=0}^{N} p_m |m\rangle\langle m| \right)\notag\\
&-  S\left(\sum_{m=0}^{N} p_m |\sqrt{\gamma} m\rangle\langle \sqrt{\gamma}m| \right) \Bigg],
\end{align}
and by analyzing its behaviour by increasing $N$, we obtain the quantum capacity in Eq.~(\ref{eq:Qfinal}). 

For $N=1$, maximizing the right hand side of Eq.~(\ref{eq:QTrunckated}), yields the optimal probability distribution to be uniform, that is $p_0=p_1=\frac{1}{2}$, as shown in  Fig.~\ref{fig:N1} together with $Q_2$. This implies that for $N=1$ the probability distribution in Eq.~(\ref{eq:Omega2}) is optimal.
\begin{figure}
    \centering
    \includegraphics[width=\columnwidth]{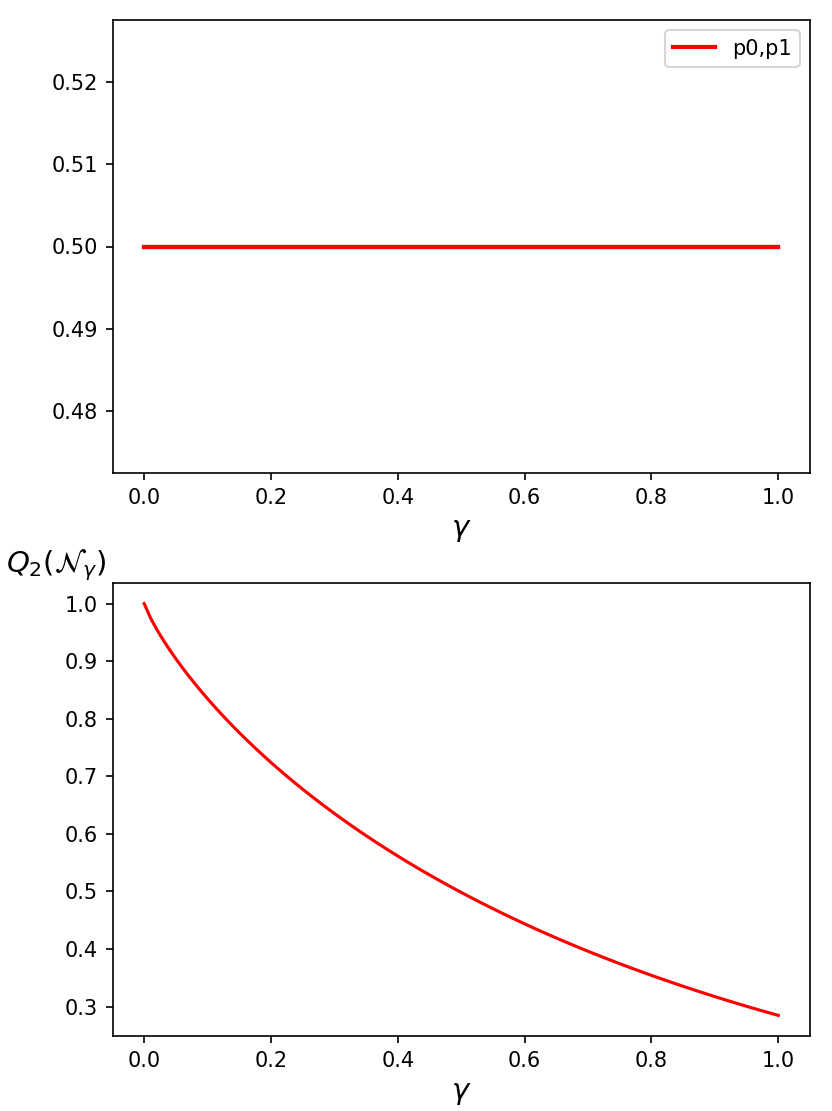}
    \caption{Top: optimal probability distribution for $N=1$ versus $\gamma$. 
    Bottom: $Q_2(\mathcal{N}_\gamma)$ versus $\gamma$.}
    \label{fig:N1}
\end{figure}
It is worth noting that even by truncating the sum in equation Eq.~(\ref{eq:Qfinal}), the numerical analysis is lengthy. The root of that goes back to the fact that by increasing $N$ not only the number of involved coherent states \eqref{eq:coherent} increases, but also their amplitudes increase. In fact, by increasing $m$, the number of required terms at the r.h.s. of \eqref{eq:coherent} to be considered increases, which is equivalent to a longer time for the numerical task. In the next section, we explain an algorithm which mitigates this problem. 


\subsection{Replica method}
\label{subsec:Replica}

We now explain our approach for numerical calculation of $Q_{N+1}$ in Eq.~(\ref{eq:QTrunckated}). Obviously, computing the first term is straightforward. For computing the second term, we will make use of the replica method \cite{GKKC,GKKC2,Calabrese}.

It is known that the von-Neumann entropy of a density matrix $\Omega$ can be written as \cite{GKKC}
\begin{equation}
\label{eq:vonNeu}
S(\Omega)=-\mathrm{Tr}(\Omega \log\Omega)=-\partial_{n}\mathrm{Tr}(\Omega^{n})|_{n=1},
\end{equation}
with $\partial_n$ denoting derivative with respect to $n$.
Therefore, instead of diagonalizing $\Omega$, one can compute the entropy through the trace of $\Omega^n$. For our purpose, we denote the density matrix appearing in the second term of $Q_{N+1}(\mathcal{N}_\gamma)$ in Eq.~(\ref{eq:QTrunckated}) which is a unitary conjugate of the complemnetary channel's output, by
\begin{equation}
    \label{eq:Omega}
    \Omega\equiv\sum_{m=0}^{N}p_{m}\ket{m\sqrt{\gamma}}\bra{m\sqrt{\gamma}},
\end{equation}
and for arbitrary $n$, express $\Omega^n$ in terms of coherent states as
\begin{equation}
    \label{eq:OmegaN}
    \Omega^n=\sum_{i,j=1 }^NC_{ij}^{(n)}\ket{\sqrt{\gamma}i}\bra{\sqrt{\gamma}j},
\end{equation}
with $C_{ij}^{(1)}=p_{i}\delta_{i,j}$. It then follows that
\begin{equation}
\label{eq:TrOmegaN}
    \mathrm{Tr}\left(\Omega^{n}\right)=\sum_{i,j=1 }^NC_{ij}^{(n)}e^{-\frac{\gamma}{2}(i-j)^2}.
\end{equation}
By considering that  $\Omega^n=\Omega^{n-1}\Omega$ and taking into account Eqs.~(\ref{eq:Omega}) and (\ref{eq:OmegaN}), the following  recurrence relation can be derived:
\begin{equation}
\label{eq:recurrence}
    C^{(n)}=C^{(n-1)}A,
\end{equation}
with 
\begin{equation}
\label{eq:A}
    A_{ij}=e^{-\frac{\gamma}{2}(i-j)^2}p_j,\qquad i,j=1,\ldots, N.
\end{equation}
Thus using  Eq.~(\ref{eq:recurrence}) in Eq.~(\ref{eq:TrOmegaN}) we conclude that
\begin{equation}
\label{eq:TrOmegaA}    
\mathrm{Tr}\left(\Omega^{n}\right)=\mathrm{Tr}\left(A^n\right)=\sum_{i=1}^Na_i^n,
\end{equation}
with $\{a_i\}_i$ the eigenvalues of the matrix $A$. Finally, from Eqs.~(\ref{eq:vonNeu}) and (\ref{eq:TrOmegaA}), we have 
\begin{equation}
\label{eq:SA}
    S(\Omega)=-\partial_{n}\mathrm{Tr}(A^{n})|_{n=1}=-\sum_{i=1}^Na_{i}\log a_{i},
\end{equation}
which implies that the numerical computation of $S(\Omega)$ can be done through the $N\times N$ matrix $A$ without any need to involve coherent states.

\medskip

By using Eq.~(\ref{eq:SA}) we compute the second term of Eq.~(\ref{eq:QTrunckated}) numerically, and optimize the whole expression over the probability distribution, $p_m$s. We find optimal values of $p_m$, as shown in Fig.~\ref{fig:OptimalP} for
$N=2,3,4,5$. Proceeding up to $N=8$, we observed the following relation between the optimal values of  $p_m$s:
\begin{align}
    &p_m<p_{m+1},\quad \;\mathrm{for} \quad 0\leq m\leq\lfloor\frac{N}{2}\rfloor,
    \label{eq:OPtimalp}\\
   &p_m=p_{N-m},\quad \mathrm{for} \quad \lfloor\frac{N}{2}\rfloor<m\leq N.
   \label{eq:OPtimalpOPtimalp}
\end{align}
\begin{figure}
    \centering
    \includegraphics[width=0.8\columnwidth]{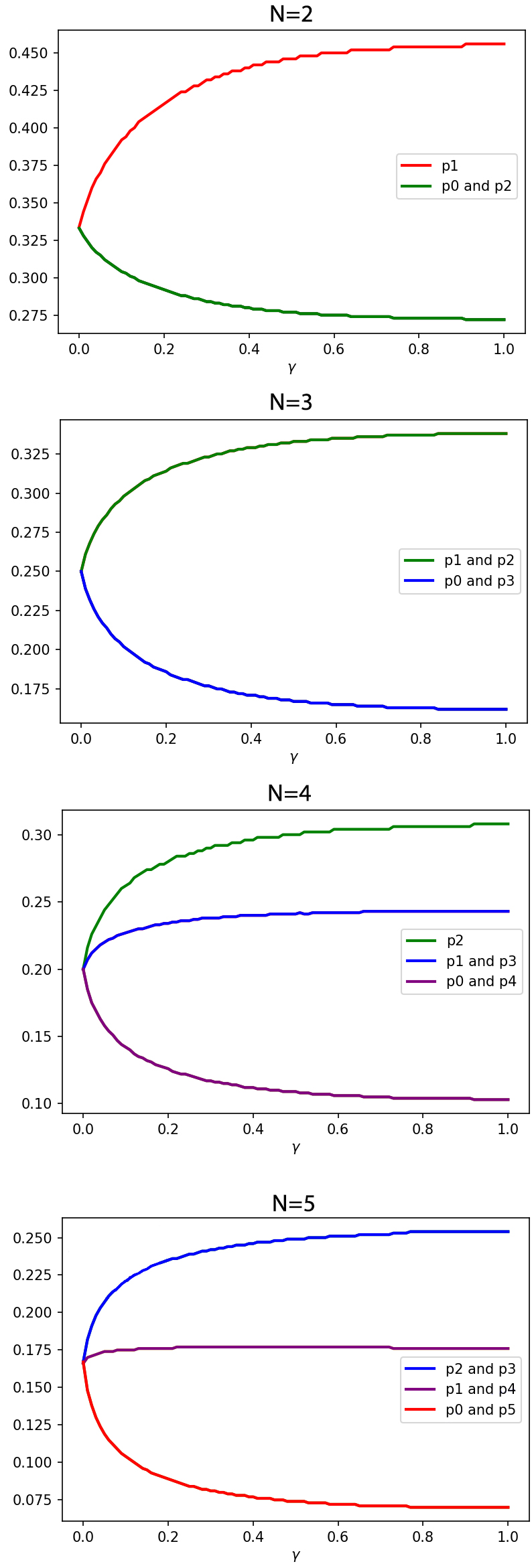}
    \caption{Optimal value of $p_m$ for $m=0,1,\cdots,N$ versus $\gamma$. From top to bottom $N=2, 3, 4, 5$.}
    \label{fig:OptimalP}
\end{figure}
For the obtained optimal probability distributions, $Q_{N+1}(\mathcal{N}_\gamma)$ is shown in Fig.~\ref{fig:Qn-N1-N5} versus $\gamma$ for  $N=1,\ldots, 8$. As expected $Q_{N+1}(\mathcal{N}_\gamma)$ monotonically decreases versus the noise parameter $\gamma$. 

For probability distribution with the pattern given in \eqref{eq:OPtimalp} and \eqref{eq:OPtimalpOPtimalp} it is straightforward to see that the mean energy of the optimal input state is $\frac{N}{2}$ which is linearly increasing by $N$. 
\begin{figure}
\centering
\includegraphics[width=\columnwidth]{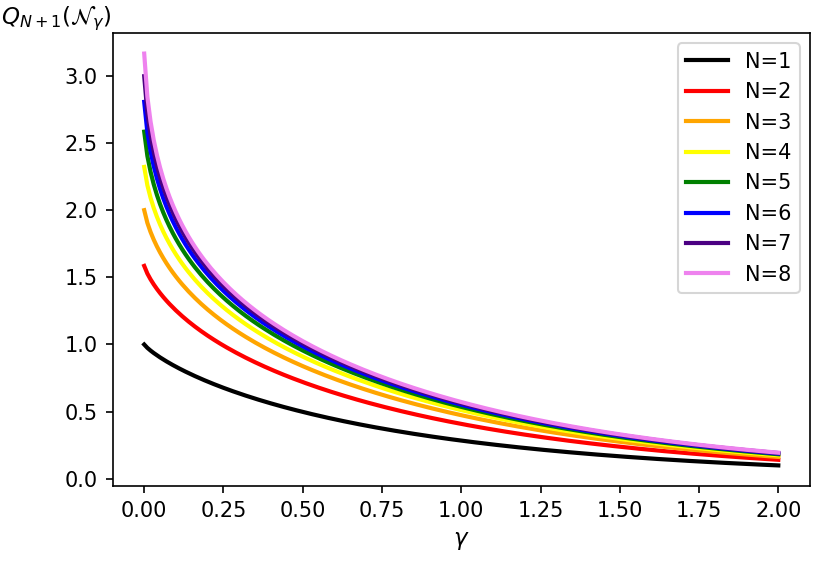}
\caption{$Q_{N+1}(\mathcal{N}_\gamma)$, as defined in Eq.~(\ref{eq:QTrunckated}),
versus $\gamma$ for $N=1,\ldots,8$.}
\label{fig:Qn-N1-N5}
\end{figure}

In the next subsection we try to figure out the probability distribution that fits well with properties in Eqs.~\eqref{eq:OPtimalp} and \eqref{eq:OPtimalpOPtimalp}.


\subsection{Optimal probability distribution}
\label{subsec:OptimalP}

We discuss here the actual form of the optimal probability distribution. 
From Eqs.~\eqref{eq:OPtimalp} and \eqref{eq:OPtimalpOPtimalp}, it is concluded that the optimal probability distribution can not have more that one peak, hence bimodal probability distributions are not optimal distributions. 
Furthermore, Eq.~\eqref{eq:OPtimalp} and \eqref{eq:OPtimalpOPtimalp} imply that the optimal probability distribution is symmetric around its peak at $m=\lfloor\frac{N}{2}\rfloor$. Therefore, 
a non-symmetric unimodal probability distribution, such as the thermal distribution, is not an acceptable candidate for optimal probability distribution in computing $Q_{N+1}$. 

A candidate for discrete probability distributions satisfying these properties is the discrete Gaussian probability distribution 
\begin{equation}
\label{eq:DisGau}
    p_m(\mu,\sigma(N,\gamma))=\frac{1}{M(\mu,\sigma)}e^{-\frac{(m-\mu)^2}{2\sigma^2(N,\gamma)}}, 
\end{equation}
with $m\in\{0,1,\cdots,N\}$. It is
centered around $\mu=\frac{N}{2}$ and has a width controlled by $\sigma$. 
Furthermore, $M(\mu,\sigma)$ is the normalization factor
\begin{equation}
    M(\mu,\sigma(N,\gamma))=\sum_{m=0}^N e^{-\frac{(m-\mu)^2}{2\sigma^2(N,\gamma)}}.
\end{equation}
From  Eqs.~\eqref{eq:OPtimalp} and \eqref{eq:OPtimalpOPtimalp} we know that for all values of 
$\gamma$, $p_m$ attains the maximum value for $m=\lfloor\frac{N}{2}\rfloor$. Therefore, 
we set $\mu=\frac{N}{2}$ and vary $\sigma$ to find the best fit to the optimal probability distribution obtained numerically in Sec.~\ref{subsec:Replica}. It is worth mentioning that for odd $N$, the maximum value of probability distribution does not pass any $p_m$, but still $p_{m}$ with $m=\lfloor\frac{N}{2}\rfloor$ and $m=\lfloor\frac{N}{2}\rfloor+1$ are equal and  have maximum values. 

By varying $\sigma$ we can fit discrete Gaussian distribution to the optimal probability distribution obtained in Sec.~{\ref{subsec:Replica}} for $N=1,\ldots,5$. We observe that $\sigma$ is linear in $N$:
\begin{equation}
    \sigma(\gamma, N)\approx a(\gamma) N+ b(\gamma),
\end{equation}
and for $\gamma>0.2$ the coefficients $a(\gamma)$ and $b(\gamma)$ are almost constant, that is $\sigma\approx 0.2N+0.6$. 

By taking $p_m$s in Eq.~(\ref{eq:QTrunckated}) from discrete Gaussian probability distribution as in Eq.~(\ref{eq:DisGau}) with $\mu=N/2$ and numerically maximizing it over $\sigma$, we calculate $Q_{N+1}(\mathcal{N}_\gamma)$. The obtained quantities for $N=1,\ldots, 5$ exactly coincide with the corresponding curves in Fig. \ref{fig:Qn-N1-N5}. Additionally, with the same procedure we obtained the behaviour of $Q_{N+1}(\mathcal{N}_\gamma)$ for $N=6,7$ and $8$ as depicted in Fig.~(\ref{fig:Qn-N1-N5}). 

As can be seen in Fig.~\ref{fig:Qn-N1-N5}, by increasing $N$, the curves become closer and closer, especially at large values of $\gamma$.  
This implies that for large values of $N$, Fig.\ref{fig:Qn-N1-N5} shows a very close approximation to the quantum capacity $Q(\mathcal{N}_\gamma)$ in
Eq.~(\ref{eq:Qfinal}) versus noise parameter $\gamma$. 
This is also reminiscent of the fact that whenever the coherent information of a one-mode
Gaussian channel is non-zero, its supremum is achieved for input power going to infinity \cite{Bradler}.
Fig.~\ref{fig:QNvsN} shows the behaviour of $Q_{N+1}(\mathcal{N}_\gamma)$ versus $N$ for some fixed values of noise parameter $\gamma$. Actually, it shows that $Q_{N+1}(\mathcal{N}_\gamma)$ saturates after a finite value of $N$ and the larger the noise  parameter is, 
the smaller is the value of $N$ at which
the saturation happens. 
\begin{figure}
\centering
\includegraphics[width=\columnwidth]{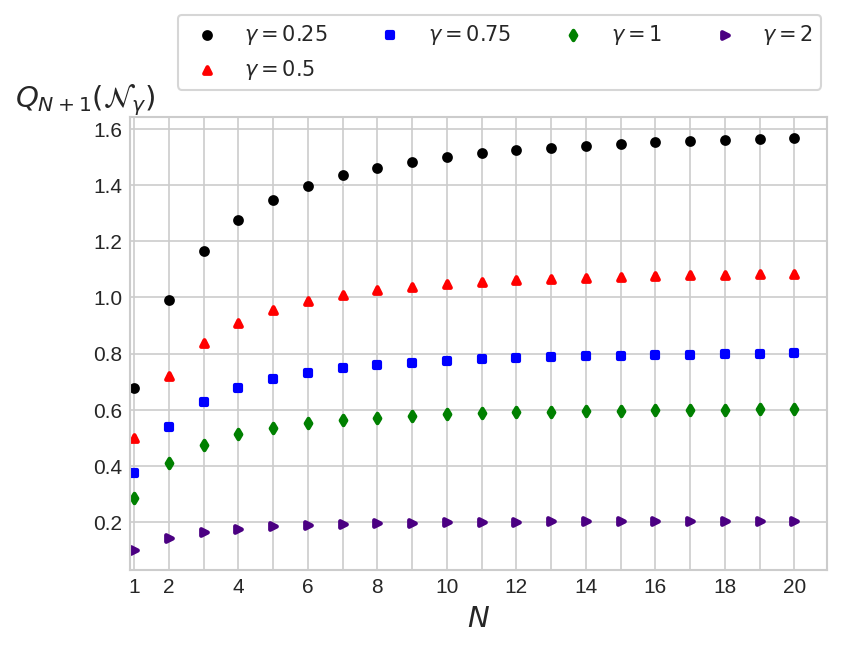}
\caption{$Q_{N+1}(\mathcal{N}_\gamma)$ as defined in Eq.~(\ref{eq:QTrunckated}) versus $N$. From top to bottom $\gamma=0.25,0.5,0.75,1,2$.}
\label{fig:QNvsN}
\end{figure}


\section{Asymptotic behaviour of quantum capacity}
\label{sec:asymptotic}

In this section, we discuss the asymptotic behavior of the quantum capacity of dephasing channel in terms of the dephasing rate, or noise parameter. As seen in Sec.~\ref{sec:QuantumCapacity}, the dephasing channel is degradable, hence its quantum capacity is equal to its private classical
capacity \cite{Smith2008}. On the other hand, the private classical capacity is always non-negative \cite{private,private2}. 
Therefore, $Q(\mathcal{N}_{\gamma})$ is always non-negative. However,
the decreasing behavior of  $Q_{N+1}(\mathcal{N}_\gamma)$ suggests that $Q_{N+1}(\mathcal{N}_\gamma)$ and hence $Q(\mathcal{N}_\gamma)$ asymptotically approaches zero from above for $\gamma\to\infty$. Actually in what follows we show that for large values of $\gamma$, $Q_{N+1}(\mathcal{N}_\gamma)$ and hence the quantum capacity decrease exponentially. 
 
While so far we have used replica method to ease the numerical analysis of the second term of $Q_{N+1}(\mathcal{N}_\gamma)$ in Eq.~(\ref{eq:QTrunckated}), here we use this technique to derive the behavior of $Q_{N+1}(\mathcal{N}_\gamma)$ and quantum capacity $Q(\mathcal{N}_\gamma)$ for large values of $\gamma$. Elements of matrix $A$ as defined in Eq.~(\ref{eq:A}) are all non-zero. Define $\epsilon\equiv e^{-\frac{\gamma}{2}}$ which is small for large values of $\gamma$. The matrix $A$ up to order $\mathcal{O}(\epsilon)$ is given by
\begin{equation}
\label{eq:AorderE}
    A_{i,j}=p_j\delta_{i,j}+\epsilon\, p_j\left(\delta_{i,j+1}+\delta_{i+1,j}\right)+\mathcal{O}(\epsilon^2).
\end{equation}
Therefore, by straightforward calculation, we obtain
\begin{equation}
  {\rm Tr}(A^n)=\sum_{m=0}^N p_m^n+\mathcal{O}(\epsilon^2),
\end{equation}
which by considering the first equality in Eq.~(\ref{eq:SA}) leads to 
$S(\Omega)=-\sum_{m=0}^Np_m\log p_m$ and therefore
$Q_{N+1}(\mathcal{N}_\gamma)$ as defined in Eq.~(\ref{eq:QTrunckated}) vanishes if we keep terms up to order $\epsilon$, because
\begin{equation}
    Q_{N+1}(\mathcal{N}_\gamma)\approx\mathcal{O}(\epsilon^2).
\end{equation}
Hence to see the asymptotic behaviour of $Q_{N+1}(\mathcal{N}_\gamma)$ for large values of $\gamma$, we write the matrix $A$ up to order $\mathcal{O}(\epsilon^2)$:
\begin{align}
    A_{i,j}&=p_j\delta_{i,j}+\epsilon p_j\left(\delta_{i,j+1}+\delta_{i+1,j}\right)\cr
    &+\epsilon^2p_j(\delta_{i,j+2}+\delta_{i+2,j})+\mathcal{O}(\epsilon^3).
\end{align}
Straightforward calculations give
\begin{equation}
  {\rm  Tr}(A^n)=\sum_{m=0}^Np_m^n+n \epsilon^2 \sum_{m=0}^{N-1}\frac{p_m^np_{m+1}-p_mp_{m+1}^n}{p_m-p_{m+1}},
\end{equation}
which, using Eq.~(\ref{eq:SA}) and replacing $\epsilon^2$ by $e^{-\gamma}$, leads to
\begin{equation}
\label{eq:QNorderEp2}
    Q_{N+1}(\mathcal{N}_{\gamma})=e^{-\gamma}\sum_{m=0}^{N-1}\frac{p_mp_{m+1}}{p_m-p_{m+1}}\log\left(\frac{p_m}{p_{m+1}}\right)+\mathcal{O}(\epsilon^3).
\end{equation}
As discussed in Sec.\ref{subsec:OptimalP} for large values of $\gamma$, the mean value and variance of optimal probability
distribution in Eq.~(\ref{eq:DisGau})
do not depend on $\gamma$. Thus, the summation in  Eq.~(\ref{eq:QNorderEp2}) does not depend 
on $\gamma$. Therefore Eq.~(\ref{eq:QNorderEp2}) implies that, for large values of $\gamma$, $Q_{N+1}(\mathcal{N}_\gamma)$ and hence $Q(\mathcal{N}_\gamma)$ approach zero exponentially.


\section{Conclusion}
\label{sec:conclusion}

Summarizing, we have studied the capability of Bosonic dephasing channel for transmitting quantum information.
We have analytically proved that for such a channel, coherent information is additive
and the optimal input state is diagonal in the Fock basis, which is invariant under the noise action.
Then, by using the replica method which makes numerical analysis technically feasible, we 
showed that the optimal probability distribution for the mixture of Fock states is unimodal and symmetric around its maximum. Among possible choices satisfying these two constraints, we took discrete Gaussian distribution as an ansatz and by varying its average and variance we showed that it fits the optimal probability distribution 
 Interestingly, this distribution is almost independent on the noise parameter
$\gamma$, but the quantum capacity varies with $\gamma$, as the output of the complementary channel depends on the noise parameter. We found it useful to truncate the dimension of Hilbert space, which is equivalent to restrict the input energy, and define $Q_{N+1}(\mathcal{N}_\gamma)$ as the maximum of coherent information in truncated space. Then, we numerically evaluated the quantum capacity by finding the asymptotic behavior of $Q_{N+1}(\mathcal{N}_\gamma)$ when enlarging the dimension of truncated Hilbert space, as it saturates to a finite value (see Figs.~\ref{fig:Qn-N1-N5} and \ref{fig:QNvsN}). Our results show that the optimal input state for transmitting quantum information through a continuous-variable quantum dephasing channel, is a mixture of number states with discrete Gaussian distribution, which is clearly not a Gaussian state. We also discussed that quantum capacity approaches zero from above when the noise parameter increases. For large values of dephasing rate, this decay is exponential.

It is worth observing that bosonic dephasing channel results an Hadamard channel as a consequence 
of having proved that its complementary
channel is entanglement breaking. 
This property implies that the triple trade-off capacity regions are single-letter, as shown in
\cite{WildeHSieh2012}.
 It could be the subject of a future investigation to determine the whole triple-trade-off region, and the replica method
might be useful there as well
(the only other known example of a bosonic channel of physical interest that is Hadamard
 and for which it is known the full triple-trade-off region,
  is the quantum-limited amplifier channel \cite{QiWilde2017}

Not only we are confident that this work can pave the way for studying quantum communication 
with continuous-variable quantum channels
beyond the usual restriction of Gaussianity, 
but it can already be useful in the context of optical communications
where dephasing effects are relevant \cite{Derickson}.
In particular, the achieved result sets an upper bound to the private communication rate, 
which is a key aspect for technological developments. 
\acknowledgments
 L. M. acknowledges financial support by Sharif University of Technology, Office of Vice President for Research
under Grant No. G930209 and hospitality by university of Camerino where parts of this work were completed.

S. M. acknowledges useful discussions with M. M. Wilde in the early stage of this project.
He also acknowledges the funding from the European Union’s Horizon 2020 research and 
innovation programme under grant agreement No 862644 (FET-Open project “QUARTET”).
\bibliography{QC-DephasingChannel}
\end{document}